%% file: main.tex
\documentclass[journal]{IEEEtran}

\ifCLASSINFOpdf
  \usepackage{graphicx}
\else
  \usepackage[dvips]{graphicx}
\fi

\usepackage{amsmath,amssymb,amsthm,mathtools}
\usepackage{mathrsfs}
\usepackage{bm}
\usepackage{dsfont}
\usepackage{extarrows}
\usepackage{textcomp,latexsym}
\usepackage{array}
\usepackage{multirow}
\usepackage{bigstrut}
\usepackage{booktabs}
\usepackage{makecell}
\usepackage{colortbl}
\usepackage{multicol}
\usepackage{algorithm}
\usepackage[noend]{algpseudocode}

\usepackage{subfigure}
\usepackage{float}
\usepackage{stfloats}
\usepackage[dvipsnames]{xcolor}
\usepackage{optidef}
\usepackage{url}
\usepackage[noadjust]{cite}
\usepackage{cleveref}
\usepackage{pifont}
\usepackage{epstopdf}
\usepackage{tikz}
\usetikzlibrary{arrows.meta,positioning,calc,shapes.geometric}
\usepackage{setspace}
\usepackage{flushend}

\definecolor{navy}{RGB}{47,85,151}
\definecolor{graybox}{RGB}{248,248,248}

\newtheorem{theorem}{Theorem}
\newtheorem{proposition}{Proposition}
\newtheorem{corollary}{Corollary}
\newtheorem{lemma}{Lemma}
\newtheorem{definition}{Definition}
\newtheorem{remark}{Remark}
\newtheorem{assumption}{Assumption}

\DeclareMathOperator{\diag}{diag}
\DeclareMathOperator{\tr}{tr}

\newcommand{\E}{\mathbb{E}}
\newcommand{\C}{\mathbb{C}}
\newcommand{\R}{\mathbb{R}}
\newcommand{\CN}{\mathcal{CN}}
\newcommand{\N}{\mathcal{N}}
\newcommand{\cQ}{\mathcal Q}
\newcommand{\cX}{\mathcal X}
\newcommand{\cZ}{\mathcal Z}
\newcommand{\cS}{\mathcal S}

\newcommand{\cW}{\mathcal W}
\newcommand{\cI}{\mathcal I}

\newcommand{\cD}{\mathcal D}
\newcommand{\cT}{\mathcal T}
\newcommand{\cV}{\mathcal V}

\newcommand{\EBD}{\mathrm{EBD}}
\newcommand{\LDEBD}{\mathrm{LD\text{-}EBD}}
\newcommand{\Prob}{\mathbb{P}}

\IEEEoverridecommandlockouts
\allowdisplaybreaks[2]

\begin{document}

\title{Event-Native Immersive Communication for 6G: From Multimodal Streams to Event Beliefs}
\author{Yiyang Zhu, Li Wei,~\IEEEmembership{Member,~IEEE} and Chau Yuen,~\IEEEmembership{Fellow,~IEEE}
\thanks{Y. Zhu and C. Yuen are with the School of Electrical and Electronics Engineering, Nanyang Technological University, Singapore 639798 (e-mail: yiyang015@e.ntu.edu.sg, chau.yuen@ntu.edu.sg).}
\thanks{L. Wei is with College of Information Science and Electronic Engineering, 
Zhejiang University, Hangzhou 310027, China (e-mail: weili\_xd@163.com).}
}

\maketitle

\begin{abstract}
The development of multimodal sensing capabilities is crucial for realizing immersive communications and enabling rich, natural remote interactions in future sixth-generation~(6G) communications. However, transmitting high-dimensional sensory streams places substantial pressure on wireless resources and leaves a fundamental question unresolved: what information about the underlying physical process should be preserved across wireless delivery? To address this issue, we establish an event-native immersive-communication framework in which a service-relevant physical event defines the communication object, while event beliefs represent the information available to transceiver agents before and after wireless delivery. Based on this event-native framework, we further propose event-belief deficiency (EBD) to evaluate quality of experience (QoE) in immersive communications, which differs from the conventional normalized mean-squared error (NMSE), as NMSE can be insensitive to physically distinct procedures whose stream distortions lie in a similar range. For Gaussian event beliefs with mode-separable delivery over Rayleigh fading, we derive the log-det EBD, its induced received signal-to-noise ratio (SNR) threshold, the exact outage probability, and the required transmit SNR. Simulation results with a physics-driven source show that stream NMSE can remain nearly unchanged across physically distinct event procedures, whereas log-det EBD reliably reflects their event-level uncertainty. Relative to a mean-squared-error (MSE)-oriented interface, event-native delivery reduces normalized log-det EBD by up to $63.6\%$ and lowers the required average received SNR by $3.47$~dB under the evaluated outage constraint, substantiating the effectiveness of the proposed framework.
\end{abstract}


\input{sections/introduction}
\input{sections/system_model}
\input{sections/metric}
\input{sections/analysis}
\input{sections/simulation}
\input{sections/conclusion}

\input{sections/appendix}

\bibliographystyle{IEEEtran}
\bibliography{IEEEabrv,references}

\end{document}

%% file: sections/introduction.tex
\section{Introduction}
\label{sec:introduction}
Immersive communication is envisioned as a key service paradigm for future sixth-generation (6G) wireless networks, with the ambition of enabling lifelike presence and real-time engagement across remote, virtual, and reconstructed environments~\cite{shen2023toward,sehad2025generative}. Different from traditional media services that are primarily designed for human-to-human content delivery, this paradigm extends the communicating entities to a broader class of intelligent agents, including humans, artificial intelligence (AI) systems, robots, and their hybrid forms. This broadened vision opens the door to a wide spectrum of applications, ranging from extended reality, holographic telepresence, remote education, and healthcare to industrial digital twins, robotic teleoperation, human--machine collaboration, and Internet-of-Senses services~\cite{vanDerHooft2023tutorial,zhao2024metaverse,dafarra2024icub3,darvish2023teleoperation,tao2024advancements,mihai2022digital}. 
Nevertheless, translating this vision into practical immersive services still requires addressing several fundamental issues.

The fundamental challenge is to determine what information should be preserved when high-dimensional multimodal observations are delivered under stringent wireless constraints. In immersive scenarios, visual, auditory, haptic, and other sensory signals are not isolated media streams, but heterogeneous observations of the same evolving physical process. A stream-centric design that directly transports these sensory signals may retain sample-level details, yet it imposes substantial bandwidth pressure, together with stringent latency and synchronization burdens. However, alleviating this pressure by reducing the transmitted content inevitably introduces a fidelity tradeoff, making the problem not merely a matter of compression, but a representation question. This issue is further complicated by cross-modal coupling: small signal-level errors may alter the perceived temporal evolution or physical consistency of an event, whereas large sample-wise deviations may be irrelevant if they do not affect the receiver's event-level understanding. Thus, the key issues in immersive communication become how to represent the underlying event and how to evaluate whether the delivered information preserves immersive fidelity.

\subsection{Related Works and Limitations}

Toward the realization of immersive communication, substantial efforts have been made across adjacent areas, each addressing a different system bottleneck. 

\emph{High-dimensional media delivery:} The most mature technical ingredients for the immersive-communication vision come from the delivery of high-dimensional media. Extended reality (XR), volumetric video, point-cloud, and holographic systems rely on rich visual or spatial representations whose transmission and rendering are constrained by bandwidth, latency, and synchronization requirements~\cite{zhu2025holocom,wang2024hermes,kurka2021bandwidth}. The corresponding literature improves this pipeline through compression, adaptive streaming, edge rendering, and synchronization, with quality evaluated by reconstruction fidelity, perceptual or rendering quality, and viewing experience. These criteria have advanced visual immersion, but they keep the communication target tied to the selected media object. If audio and haptic signals are added as time-aligned streams, the evaluation still scores the delivered streams or rendered views, rather than whether the modalities jointly support a physically consistent event.

\emph{Low-latency haptic communication:} Another indispensable ingredient comes from haptic communication and the Tactile Internet. In this setting, the main difficulty is the effect of delay, packet loss, and jitter on closed-loop physical interaction, rather than the scale of visual data~\cite{antonakoglou2018toward,promwongsa2021comprehensive,steinbach2012haptic}. Delayed or unreliable feedback can change the dynamics of a bilateral system and even destabilize the interaction loop. Existing designs therefore employ coding, prediction, and control-aware transmission for position, force, and vibration signals. The resulting notion of quality is operational, with latency, jitter, reliability, and transparency. These metrics are essential for touch-enabled interaction, but they are defined around the haptic channel itself and do not determine whether the delivered touch agrees with simultaneous visual and auditory evidence of the same event.

\emph{Task-oriented transmission:} Broader communication-theoretic ingredients are provided by semantic and task-oriented transmission. Deep joint source--channel coding (JSCC) learns source and channel mappings for a chosen feature~\cite{gunduz2025jscc,bourtsoulatze2019deep,wu2024deepjsccMIMO}, while semantic and task-oriented communication optimizes semantic similarity or downstream utility~\cite{xie2021semantic,lyu2024semantic,dai2022nonlinear}. Remote-source formulations further distinguish an underlying source from its noisy observations~\cite{witsenhausen1980,oohama2014indirect}. Their quality criteria follow the declared objective, e.g., reconstruction distortion, semantic similarity, or remote-source estimation error. While this objective-driven view protects selected features or states over imperfect wireless links, the semantic target, hidden variable, and loss function are chosen case by case, making the resulting criteria difficult to generalize across time-evolving multimodal interactions.

Together, these ingredients optimize different targets: a rendered feature, a stable feedback loop, or a task-specific utility. In an immersive interaction, however, visual, auditory, haptic, and other observations are coupled through the same time-evolving event, so preserving each target separately does not necessarily preserve the event perceived or acted upon by the receiving agent. This mismatch leaves two basic questions unanswered: \emph{What should immersive communication transmit? How should its quality be measured?}

\subsection{Main Contributions}
\label{subsec:contributions}

To answer the above questions, this paper develops a physics-informed foundation for immersive communication. The central idea is to make a service-relevant physical event, rather than its sensory manifestations, the communication object and to evaluate the event information preserved at the receiving agent. This leads to \emph{event-native immersive communication}, where event beliefs define the operational information states and event-belief deficiency (EBD) defines process-level fidelity. To our knowledge, this is the first technical framework that jointly specifies a physical-event communication object and an event-level fidelity criterion for immersive communication. Main contributions are as follows.

\begin{itemize}
\item
We establish a general system model for immersive communication. A multimodal physical process acts as the source, a service-relevant event defines the communication object, and the transceiver beliefs represent the event information available before and after wireless delivery. We prove that the belief is sufficient for event-dependent inference and show that conventional source transmission and stream reconstruction arise as special cases of the proposed model.

\item
Building on this system model, we introduce the sensing-supported agent belief as the communication reference and define EBD as the communication-induced distortion of event-dependent decision capability, which leads to an event-belief reliability criterion and a corresponding communication-design problem. For Gaussian event beliefs with mode-separable delivery over Rayleigh fading, we derive the log-determinant EBD (LD-EBD), the associated received signal-to-noise ratio (SNR) threshold, the exact outage probability, and the required transmit SNR. The same formulation also exposes how event-mode resources affect the reliability threshold through the delivered precision.

\item
We validate the proposed system, metric, and reliability analysis in a physics-based environment. The results show that stream normalized mean-squared error (NMSE) can be nearly insensitive to physically distinct event procedures, whereas LD-EBD accurately reflects their event-level uncertainty. The analytical outage expression closely matches Monte Carlo results. Relative to a mean-squared-error (MSE)-oriented interface, event-native delivery reduces normalized LD-EBD by up to $63.6\%$ and lowers the required average received SNR by $3.47$~dB under the evaluated LD-EBD outage constraint.
\end{itemize}

\subsection{Organization and Notation}

The remainder of this paper is organized as follows. Section II develops the event-native system model, and Section III introduces EBD as the receiver quality of experience (QoE) metric. Section IV analyzes event-belief reliability over Rayleigh fading. Section V presents the physics-based numerical results, followed by the conclusion in Section VI.

\textit{Notation:} Lowercase and uppercase boldface letters denote vectors and matrices, respectively, while calligraphic letters denote sets or spaces. For Hermitian matrices $\mathbf A$ and $\mathbf B$, $\mathbf A\succeq\mathbf B$ means that $\mathbf A-\mathbf B$ is positive semidefinite. The symbols $\mathbb R$, $\mathbb C$, $\mathbb P(\cdot)$, $\mathbb E\{\cdot\}$, and $\perp\!\!\!\perp$ denote the real field, complex field, probability measure, expectation, and (conditional) independence, respectively. Moreover, $\tr(\cdot)$ is the trace operator, $\diag(\cdot)$ forms a diagonal matrix from its arguments, and $\mathbf I_L$ is the $L\times L$ identity matrix. The symbol $\mathbf 0_L$ denotes an all-zero vector or matrix, with dimensions clear from context. Finally, $\N(\bm\mu,\boldsymbol\Sigma)$ and $\CN(\bm\mu,\boldsymbol\Sigma)$ denote real and circularly symmetric complex Gaussian distributions, respectively, and all random elements are defined on standard Borel spaces so that regular conditional distributions exist.

%% file: sections/system_model.tex
\section{System Model: From a Physical Process to Receiver Knowledge}
\label{sec:system_model}

This section models how event-related knowledge is formed from a remote physical process and delivered over a wireless link. 
We first establish the causally evolving physical model, then characterize transmitter-side event inference, and finally specify event-native encoding, wireless transmission, and receiver-side inference.

\subsection{System Overview}

As shown in Fig.~\ref{fig:system}, a transmitter-side agent (Tx-A) equipped with $M$ synchronized heterogeneous sensors observes a remote physical environment over an observation window and communicates with a receiver-side agent (Rx-A) over a block-fading channel. During each window, the physical process evolves and generates multimodal sensor records available to the Tx-A, whereas the Rx-A observes only the received signal and channel state. From the multimodal records, the Tx-A forms a posterior belief over the service-relevant event and generates the transmitted signal based on this belief. The Rx-A forms its own posterior belief over the same event from the wireless evidence and uses it for subsequent inference or decision. 

\begin{figure}[t]
    \centering
    \includegraphics[width=\linewidth]{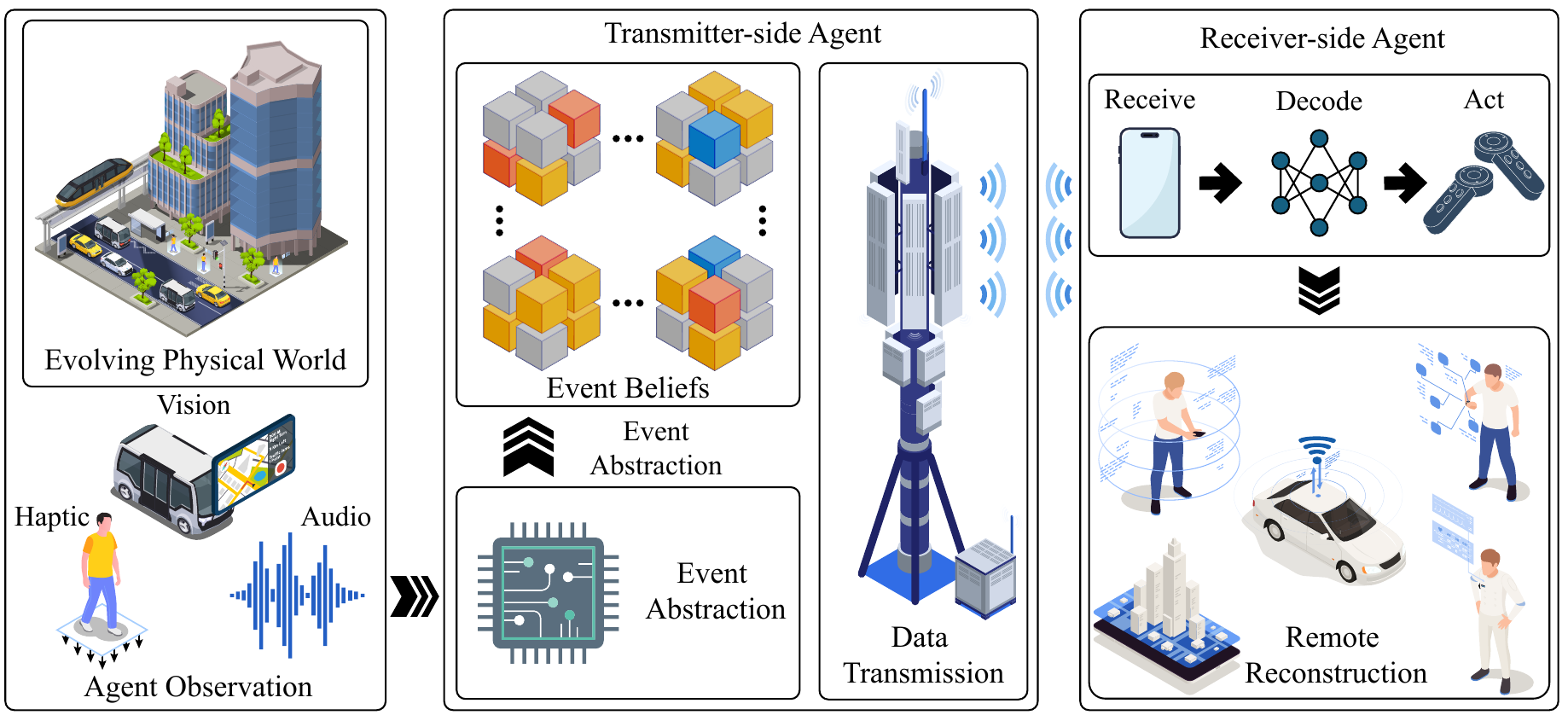}
    \caption{System architecture of the proposed event-native immersive communication system.}
    \label{fig:system}
\end{figure}

\subsection{Physical Process and Multimodal Sensing}

We first characterize the physical process that gives rise to the multimodal evidence available at the Tx-A. The $M$ sensors are synchronized over each observation window and provide complementary measurements of the same physical evolution. Following~\cite{brunton2024mlpde,karniadakis2021physicsml}, we represent the remote environment by a spatiotemporal field $\Xi:\cD\times\cT\to\R^{d_\Xi}$ whose evolution satisfies
\begin{align}
  &\mathcal L_{\boldsymbol\theta}[\Xi](\bm r,t)
  =
  \bm f_{\rm ext}(\bm r,t)+\bm\omega(\bm r,t),
  &&
  (\bm r,t)\in\cD\times\cT,
  \label{eq:physical_source_pde}\\
  &\mathcal B_{\boldsymbol\theta}[\Xi](\bm r,t)
  =
  \bm 0,
  &&
  \bm r\in\partial\cD,
  \label{eq:physical_source_bc}\\
  &\Xi(\bm r,t_0)
  =
  \Xi_0(\bm r),
  &&
  \bm r\in\cD.
  \label{eq:physical_source_ic}
\end{align}
Here, $\bm r$ and $t$ denote the spatial and temporal coordinates, respectively. The parameterized operators $\mathcal L_{\boldsymbol\theta}$ and $\mathcal B_{\boldsymbol\theta}$ describe the physical evolution and boundary behavior and may represent coupled nonlinear partial differential equations together with auxiliary lumped dynamics. The external excitation $\bm f_{\rm ext}$, unresolved disturbance $\bm\omega$, and initial condition $\Xi_0$ account for the variability of the physical environment and jointly induce the field law $P_\Xi$. Let $\cV\subseteq\cD$ denote the region of interest and let $T_{\rm e}$ denote the observation duration. Over the spatial--temporal window $\cW=\cV\times[t_0,t_0+T_{\rm e}]$, the corresponding physical episode is
\begin{align}
  \Xi_{\cW}
  &=
  \mathbf R_{\cW}\Xi,
  \label{eq:physical_window}
\end{align}
where $\mathbf R_{\cW}$ restricts the field to $\cW$. The episode $\Xi_{\cW}$ contains the complete physical evolution within the window, while the Tx-A accesses it only through its sensors. Specifically, sensor $m$, with sensing support $\cV_m\subseteq\cV$, produces at time $t\in[t_0,t_0+T_{\rm e}]$ the causal record
\begin{align}
  X^{(m)}_{[t_0,t]}
  &=
  \mathcal H_{m,t}\!\left(
    \mathbf R_{\cV_m\times[t_0,t]}\Xi,
    V^{(m)}_{[t_0,t]}
  \right),
  \label{eq:sensor_operator}
\end{align}
where $\mathcal H_{m,t}$ is the sensing operator and $V^{(m)}$ collects sensor noise and unmodeled transduction effects. Since $\mathcal H_{m,t}$ depends only on the field history up to $t$, \eqref{eq:sensor_operator} explicitly captures sensing causality. At the end of the observation, define
\begin{align}
  X^{(m)}
  &=
  X^{(m)}_{[t_0,t_0+T_{\rm e}]},
  &
  X^{1:M}
  &=
  \bigl(X^{(1)},\ldots,X^{(M)}\bigr).
  \label{eq:complete_sensor_records}
\end{align}
Collectively, the sensing operators jointly induce the physical observation law
\begin{align}
  X^{1:M}
  &\sim
  P_{X^{1:M}|\Xi_{\cW}}(\cdot|\Xi_{\cW}).
  \label{eq:physical_sensing_law}
\end{align}
This joint law captures the spatial, temporal, and cross-modal dependence induced by the common physical process. The resulting record $X^{1:M}$ constitutes the Tx-A evidence about the remote physical episode.
To determine which information in this evidence is relevant for communication, we next define the event associated with the intended service.

\subsection{Service-Relevant Event and Transmitter-Side Inference}

The physical episode $\Xi_{\cW}$ generally contains information beyond that required by the immersive service. Accordingly, two episodes are regarded as representing the same service-relevant event if they are indistinguishable in distribution under all declared interpretations~\cite{liu2022indirect}. To formalize this notion, let $\mathfrak A$ index a family of interpretation targets, where $Z_a$ denotes the physical quantity of interest associated with interpretation $a\in\mathfrak A$.

\begin{definition}[Service-relevant event]
For two physical episodes $\xi_{\cW}$ and $\xi'_{\cW}$, $\forall a\in\mathfrak A$, define
\begin{align}
  \xi_{\cW}\sim_{\mathfrak A}\xi'_{\cW}
  \quad\Longleftrightarrow\quad
  P_{Z_a|\Xi_{\cW}}(\cdot|\xi_{\cW})
  &=
  P_{Z_a|\Xi_{\cW}}(\cdot|\xi'_{\cW}).
  \label{eq:event_equivalence}
\end{align}
The service-relevant event is the corresponding equivalence class
\begin{align}
  Q
  &=
  [\Xi_{\cW}]_{\sim_{\mathfrak A}}
  =
  \mathcal T_{\rm evt}(\Xi_{\cW})
  \in\cQ,
  \label{eq:event_definition}
\end{align}
where $\cQ$ is the set of equivalence classes and $\mathcal T_{\rm evt}$ is the associated event map.
\end{definition}
Episodes mapped to the same $Q$ may differ in their trajectories and sensor records, yet induce identical conditional laws for all $a\in\mathfrak A$. Hence, up to bijective relabeling, $Q$ is the coarsest deterministic function of $\Xi_{\cW}$ that preserves these laws, thereby providing a common decision-relevant state space for comparing the information available at the two agents. Fix $a\in\mathfrak A$ and write $Z=Z_a$. The corresponding conditional law is
\begin{align}
  Z = \mathcal G_{\rm imm}(Q,E),
  \quad
  Z \sim P_{\rm imm}(\cdot|Q),
  \label{eq:immrep}
\end{align}
where $E$ represents residual task-specific variation. The event $Q$ specifies the aspects of the physical episode relevant to the selected interpretation. However, it is not directly observed by the Tx-A, but inferred from $X^{1:M}$.

To obtain the likelihood required for this inference, we combine the event map with the physical sensing law. Since $\mathcal T_{\rm evt}$ can be many-to-one, the detailed physical episode remains random when $Q$ is given. Marginalizing this within-event variability in \eqref{eq:physical_sensing_law} yields
\begin{align}
  &P_{X^{1:M}|Q}(dx^{1:M}|q)
  \nonumber\\
  &\quad=
  \int
  P_{X^{1:M}|\Xi_{\cW}}
  (dx^{1:M}|\xi_{\cW})
  P_{\Xi_{\cW}|Q}
  (d\xi_{\cW}|q).
  \label{eq:physical_to_sensor_likelihood}
\end{align}
This likelihood captures within-event variability and joint sensing effects. Assuming conditional independence of the interpretation and sensing residuals given $Q$, we obtain
\begin{align}
  Z
  &\mathrel{\perp\!\!\!\perp}
  X^{1:M}
  \mid Q.
  \label{eq:zx_condind}
\end{align}
The prior $P_Q$ and the likelihood in \eqref{eq:physical_to_sensor_likelihood} then determine the Tx-A posterior event belief
\begin{align}
  B_T(dq)
  &=
  P_{Q|X^{1:M}}(dq|X^{1:M}),
  \label{eq:txbelief}
\end{align}
which is a probability-measure-valued random variable on $\cQ$. When the corresponding densities exist, Bayes' rule gives
\begin{align}
  b_T(q;x^{1:M})
  &=
  \frac{
    p_{X^{1:M}|Q}(x^{1:M}|q)p_Q(q)
  }{
    \displaystyle
    \int_{\cQ}
    p_{X^{1:M}|Q}(x^{1:M}|q')p_Q(q')\,dq'
  }.
  \label{eq:txbayes}
\end{align}
The belief $B_T$ captures both the event information supported by the sensor observations and the uncertainty remaining after multimodal sensing. Under \eqref{eq:zx_condind}, it determines the posterior predictive law of $Z$ as
\begin{align}
  P_{Z|X^{1:M}}(dz|X^{1:M})
  &=
  \int_{\cQ}
  P_{\rm imm}(dz|q)B_T(dq).
  \label{eq:z_from_bt}
\end{align}
This predictive relation leads to the following sufficiency result.

\begin{lemma}[Event-belief sufficiency]
\label{prop:belief_sufficiency}
Under \eqref{eq:immrep}, \eqref{eq:physical_to_sensor_likelihood}, and \eqref{eq:zx_condind},
\begin{align}
  Z
  &\mathrel{\perp\!\!\!\perp}
  X^{1:M}
  \mid B_T.
  \label{eq:belief_sufficiency}
\end{align}
Hence, for any Bayesian decision problem with loss $\ell(\hat z,Z)$, an optimal decision rule based on $X^{1:M}$ can be expressed as a function of $B_T$.
\end{lemma}

\begin{proof}
See Appendix~\ref{app:proof_belief_sufficiency}.
\end{proof}

Lemma~\ref{prop:belief_sufficiency} justifies using $B_T$ as the Tx-A information state for event-dependent encoding.

\subsection{Event-Native Encoding and Wireless Transmission}

Based on its posterior event belief, the Tx-A generates an event representation $S$ and a length-$n$ channel input $\bm U$ according to
\begin{align}
  S
  &\sim
  \psi(\cdot|B_T),
  \bm U \sim \phi(\cdot|S),
  \frac{1}{n} \E\!\left[\lVert\bm U\rVert^2\right]\le P,
  \label{eq:embedding}
\end{align}
where $\psi$ is the event representation mapping, $\phi$ is the channel encoder, $P$ is the average transmit-power constraint, and the intermediate representation $S$ separates event representation from channel encoding, respectively.

\begin{remark}[Stream-reconstruction specialization]
If the service requires reconstruction of the full multimodal record and sensing is deterministic given the physical episode, the event map preserves all sensor-domain distinctions. Thus, up to bijective relabeling, $Q=X^{1:M}$, and
\begin{align}
B_T(dq)
&=
\delta_{X^{1:M}}(dq),
&
S
&\sim
\psi(\cdot|X^{1:M}).
\label{eq:stream_specialization}
\end{align}
The receiver therefore reconstructs $X^{1:M}$, recovering stream-oriented transmission as a special case.
\label{rem:stream_event}
\end{remark}

For the general event-native model, the Tx-A transmits $\bm U$ over a block-fading channel whose coefficient $H$ remains constant during the $n$ channel uses. The received block is
\begin{align}
  \bm Y
  &=
  H\bm U+\bm W,
  &
  \bm W
  &\sim
  \CN(\mathbf 0_n,N_0\mathbf I_n).
  \label{eq:channel_model}
\end{align}
Given the channel input and channel state, the received block is independent of all upstream variables; equivalently,
\begin{align}
  P_{\bm Y|
    \bm U,H,\Xi_{\cW},Q,X^{1:M},B_T,S,Z}
  &=
  P_{\bm Y|\bm U,H}.
  \label{eq:physical_channel_markov}
\end{align}
The pair $(\bm Y,H)$ is therefore the wireless evidence available to the Rx-A. Its event belief must be inferred from the observation model induced by the physical sensing process, the Tx-A encoder, and the wireless channel.

\subsection{Receiver-Side Event Inference}

For notational convenience, let $B_T[x^{1:M}]=P_{Q|X^{1:M}}(\cdot|x^{1:M})$ denote the Tx-A belief associated with a realization $x^{1:M}$. Conditioned on an event realization $q$ and channel state $h$, the received-signal law is obtained by averaging over the multimodal observations, event representations, and channel inputs:
\begin{align}
  &P_{\bm Y|Q,H}(d\bm y|q,h)
  \nonumber\\
  =&\int_{\cX}
  \int_{\cS}
  \int_{\C^n}
  P_{X^{1:M}|Q}(dx^{1:M}|q)
  \psi\!\left(ds\middle|B_T[x^{1:M}]\right)
  \nonumber\\
  &\times
  \phi(d\bm u|s)
  P_{\bm Y|\bm U,H}(d\bm y|\bm u,h),
  \label{eq:induced_likelihood}
\end{align}
where $\cX$ and $\cS$ denote the spaces of $X^{1:M}$ and $S$, respectively. This induced likelihood combines the event-conditioned sensing law with the Tx-A encoding rule and the wireless channel. Combining \eqref{eq:induced_likelihood} with the prior $P_Q$, the Rx-A forms the posterior event belief
\begin{align}
  B_R(dq)
  &=
  P_{Q|\bm Y,H}(dq|\bm Y,H).
  \label{eq:rxbelief}
\end{align}
When the corresponding densities exist, we obtain
\begin{align}
  b_R(q;\bm y,h)
  &=
  \frac{
    p_{\bm Y|Q,H}(\bm y|q,h)p_Q(q)
  }{
    \displaystyle
    \int_{\cQ}
    p_{\bm Y|Q,H}(\bm y|q',h)p_Q(q')\,dq'
  }.
  \label{eq:rxbayes}
\end{align}
Thus, $B_R$ is the event belief inferred from the wireless evidence through the end-to-end observation model, rather than reconstructed from $B_T$. Using the same law as the Tx-A, the Rx-A obtains the posterior predictive law
\begin{align}
  P_{Z|\bm Y,H}(dz|\bm Y,H)
  &=
  \int_{\cQ}
  P_{\rm imm}(dz|q)B_R(dq).
  \label{eq:z_from_br}
\end{align}
Based on $B_R$, it then produces an interpretation or decision according to
\begin{align}
  \hat Z
  &\sim
  \pi_{\rm R}(\cdot|B_R),
  \label{eq:rx_decision}
\end{align}
where $\pi_{\rm R}$ denotes the receiver-side inference rule. The beliefs $B_T$ and $B_R$ are posteriors on the same event space but are conditioned on different evidence, and thus induce two information structures for the same family of event-dependent decision problems, providing the basis for the comparison.

%% file: sections/metric.tex
\section{Event-Belief Deficiency: From Decision Risk to Wireless Reliability}
\label{sec:metric_problem}

Immersive communication supports interpretation and decision making about a remote physical environment. Its service performance reflects both the uncertainty remaining after multimodal sensing and the additional degradation introduced by representation and wireless delivery. A communication metric should therefore isolate the latter relative to the information actually available at the Tx-A. Accordingly, in this section we first characterize the decision value of an information state, then establish the Tx-A communication reference, and finally define event-belief deficiency, its reliability criterion, and the associated resource-design problem.

\subsection{From Stream Distortion to Event-Dependent Decision Risk}

A stream-oriented system commonly evaluates the recovered modality records through
\begin{align}
  D_{\rm str}
  &=
  \sum_{m=1}^{M}
  w_m
  \E\!\left[
    d_m\!\left(
      X^{(m)},
      \hat X^{(m)}
    \right)
  \right],
  \label{eq:streammetric}
\end{align}
where $d_m$ and $w_m$ are the distortion measure and weight of modality $m$, respectively. This criterion is appropriate when the reconstructed records are the desired receiver outputs. In immersive communication, however, the records provide intermediate evidence about a remote physical event. Their service value is therefore determined by the event-dependent decisions they enable, which need not be characterized by sensor-domain distortion.

Fix an interpretation $a\in\mathfrak A$ with quantity of interest $Z_a$, and write $Z=Z_a$. Let $\mathsf A$ denote the action space and let $d:\cZ\times\mathsf A\to\R_+$ denote the associated loss.

\begin{definition}[Event-dependent decision risk]
For an information state $\cI$, the minimum decision risk for the task $(Z,d)$ is
\begin{align}
  \mathcal R_{Z,d}(\cI)
  &=
  \inf_{\eta}
  \E\!\left[
    d\!\left(
      Z,
      \eta(\cI)
    \right)
  \right],
  \label{eq:cognition_risk}
\end{align}
where the infimum is over all admissible, possibly randomized decision rules with actions in $\mathsf A$.
\end{definition}

The risk $\mathcal R_{Z,d}(\cI)$ is the best achievable performance for the declared task given $\cI$, independently of a particular inference implementation. When the task is fixed, we write $\mathcal R(\cI)$.

\begin{proposition}[Stream distortion does not determine event risk]
\label{prop:stream_event_nonequivalence}
There exist two Rx-A information states with the same minimum stream-reconstruction mean-squared error but different minimum squared-error risks for the same event-dependent quantity $Z$.
\end{proposition}

\begin{proof}
See Appendix~\ref{app:proof_stream_event_nonequivalence}.
\end{proof}

Proposition~\ref{prop:stream_event_nonequivalence} shows that sensor-domain reconstruction quality does not generally determine the decision capability available for an immersive task. Conversely, a task-specific reconstruction objective remains included as a particular choice of $Z$, $\mathsf A$, and $d$. Decision risk therefore provides a common performance measure for comparing the event information available before and after wireless delivery. This comparison first requires an appropriate pre-communication reference.

\subsection{Communication Reference and Risk Decomposition}
A communication-induced loss should be defined relative to the information available prior to transmission. Direct access to the latent event $Q$ provides an oracle reference, thereby excluding uncertainty due to multimodal sensing. We therefore use the Tx-A belief $B_T$, inferred from the sensor observations, as the pre-communication reference. Specifically, we consider three information states:
\begin{align}
  \cI_{\rm ref}
  &=
  B_T,
  &
  \cI_{\rm rx}
  &=
  (\bm Y,H),
  &
  \cI_0
  &=
  \varnothing,
  \label{eq:information_states}
\end{align}
where $\cI_0=\varnothing$ denotes the absence of event-window observations, while the event prior and task model remain known. Their corresponding decision risks are
\begin{align}
  D_{\rm ref}
  &=
  \mathcal R(B_T),
  &
  D_{\rm rx}(\psi,\phi)
  &=
  \mathcal R(\bm Y,H),
  &
  D_0
  &=
  \mathcal R(\varnothing).
  \label{eq:cognition_risks}
\end{align}
The prior risk $D_0 $is the minimum risk attainable without window-specific observations. The reference risk $D_{\rm ref}$ is the minimum risk supported by multimodal sensing, prior to representation and wireless transmission, whereas the receiver risk $D_{\rm rx}$ is the corresponding risk after the complete communication chain. In general, $D_{\rm ref}>0$ because $B_T$ remains subject to sensing noise, limited coverage, and unresolved physical variability. By Lemma~\ref{prop:belief_sufficiency},

\begin{align}
  \mathcal R(X^{1:M}) = \mathcal R(B_T),
  \quad
  \mathcal R(\bm Y,H) = \mathcal R(B_R).
  \label{eq:belief_risk_sufficiency}
\end{align}
Hence, the comparison in \eqref{eq:cognition_risks} is equivalently between the event beliefs available at the Tx-A and Rx-A. Given $B_T=b$, the representation mapping, channel encoder, and wireless channel induce the end-to-end kernel
\begin{align}
  &K_{\psi,\phi}(d\bm y,dh|b)\nonumber\\
  =
  &P_H(dh)
  \int_{\cS}
  \psi(ds|b)
  \int_{\C^n}
  \phi(d\bm u|s)
  P_{\bm Y|\bm U,H}(d\bm y|\bm u,h).
  \label{eq:communication_kernel}
\end{align}
Consequently,
\begin{align}
  P_{\bm Y,H|Q}(d\bm y,dh|q)
  &=
  \int
  K_{\psi,\phi}(d\bm y,dh|b)
  P_{B_T|Q}(db|q).
  \label{eq:blackwell_kernel}
\end{align}
Equation~\eqref{eq:blackwell_kernel} implies that the wireless evidence is generated from $B_T$ through a stochastic kernel conditionally independent of $Q$. Hence, any Rx-A decision rule can be replicated from $B_T$ by sampling a virtual receiver observation according to $K_{\psi,\phi}$ and applying the same rule.

\begin{proposition}[Communication-induced risk ordering]
\label{prop:risk_ordering}
For every event-dependent task $(Z,d)$ with finite risk,
\begin{align}
  D_{\rm ref}
  &\le
  D_{\rm rx}(\psi,\phi)
  \le
  D_0.
  \label{eq:distorder}
\end{align}
\end{proposition}

\begin{proof}
The simulation argument above gives $D_{\rm ref}\le D_{\rm rx}$ and is formalized by Blackwell's comparison of statistical experiments~\cite{blackwell1953,goela2020channel}. The inequality $D_{\rm rx}\le D_0$ follows because the Rx-A can ignore its observation and apply the prior-optimal rule.
\end{proof}

The ordering in \eqref{eq:distorder} separates the decision value supplied by sensing from the degradation introduced by communication. In particular, $D_0-D_{\rm ref}$ is the task-specific decision improvement provided by multimodal sensing, whereas $D_{\rm rx}-D_{\rm ref}$ is the portion of that capability lost through representation and wireless delivery. 

\subsection{Event-Belief Deficiency}

\begin{definition}[Event-belief deficiency]
For an event-dependent task $(Z,d)$ and mappings $(\psi,\phi)$, the event-belief deficiency is
\begin{align}
  \EBD_{Z,d}(\psi,\phi)
  &=
  \mathcal R_{Z,d}(B_R)
  -
  \mathcal R_{Z,d}(B_T)
  \nonumber\\
  &=
  D_{\rm rx}(\psi,\phi)
  -
  D_{\rm ref}.
  \label{eq:EBD}
\end{align}
\end{definition}
EBD quantifies the decision degradation introduced by the communication chain after multimodal sensing. For fixed physical and sensing models, $D_{\rm ref}$ is communication-independent, whereas $D_{\rm rx}$ depends on the representation, encoding, and wireless channel. Thus, EBD isolates the communication-induced loss relative to the event information available at the Tx-A. We omit the task index $(Z,d)$ when clear from context. From \eqref{eq:distorder},
\begin{align}
  0
  &\le
  \EBD(\psi,\phi)
  \le
  D_0-D_{\rm ref}.
  \label{eq:ebd_bounds}
\end{align}
The lower bound follows because communication cannot outperform direct access to $B_T$, while the upper bound corresponds to complete loss of the sensing-derived decision value. Replacing the receiver evidence with a stochastically more informative state cannot increase EBD. For a given task, $\EBD=0$ indicates that the Rx-A achieves the same minimum risk as the Tx-A, without requiring $B_R=B_T$ realization-wise. If EBD vanishes for all tasks in a specified decision family, the two information states are decision-equivalent over that family. Since EBD depends on the quantity of interest and the loss function, we normalize it by the total decision value provided by sensing. For $D_0>D_{\rm ref}$, define

\begin{align}
  \overline{\EBD}_{Z,d}(\psi,\phi)
  &=
  \frac{
    D_{\rm rx}(\psi,\phi)-D_{\rm ref}
  }{
    D_0-D_{\rm ref}
  }.
  \label{eq:normalizedEBD}
\end{align}
The normalized EBD is the fraction of the sensing-supported decision value lost through communication. Its endpoints correspond to complete preservation of that value and no decision gain over the prior, respectively. 

\begin{corollary}[Log-loss representation of EBD]
\label{cor:log_ebd}
Let the receiver action be a predictive distribution $\hat p_Z$ and
\begin{align}
  d_{\log}(z,\hat p_Z)
  &=
  -\log \hat p_Z(z).
  \label{eq:log_loss}
\end{align}
Whenever the corresponding quantities are finite,
\begin{align}
  \EBD_{\log}(\psi,\phi)
  &=
  I\!\left(
    Z;
    B_T
    \mid
    \bm Y,H
  \right)
  \nonumber\\
  &=
  \mathsf H(Z|\bm Y,H)
  -
  \mathsf H(Z|B_T),
  \label{eq:ebd_information}
\end{align}
where $\mathsf H$ denotes entropy or differential entropy, as appropriate.
\end{corollary}

\begin{proof}
Under logarithmic loss, the minimum decision risk is the corresponding conditional entropy~\cite{cover2006elements}. Since the communication chain implies
$Z\mathrel{\perp\!\!\!\perp}(\bm Y,H)\mid B_T$,
we have
$\mathsf H(Z|B_T,\bm Y,H)=\mathsf H(Z|B_T)$, which yields
\eqref{eq:ebd_information}.
\end{proof}

The log-loss EBD is the information about $Z$ available in the Tx-A belief but absent from the wireless evidence. 
The definitions above characterize average degradation over the operating conditions; reliability additionally requires control of large state-dependent deficiencies.

\subsection{Event-Belief Reliability and Resource Design}
Average EBD may obscure severe degradation under unfavorable operating conditions. Let $\Theta$ denote the state used to stratify service performance, such as the channel power gain or a prescribed context. To preserve the communication-induced interpretation, $D_{\rm ref}$ and $D_{\rm rx}$ are evaluated under the same realization. Here, $\Theta$ is only a conditioning variable and does not augment the agents’ information; thus, the decision rule remains a function of $\cI$. Any operating variable available to the decision maker must instead be included in $\cI$. For $\Theta=\theta$, define the conditional decision risk as
\begin{align}
  \mathcal R_{\theta}(\cI)
  &=
  \inf_{\eta}
  \E\!\left[
    d\!\left(
      Z,
      \eta(\cI)
    \right)
    \,\middle|\,
    \Theta=\theta
  \right],
  \label{eq:conditional_risk}
\end{align}
where the infimum is over the same class of admissible decision rules as in \eqref{eq:cognition_risk}. The corresponding state-dependent EBD is
\begin{align}
  \EBD_{\theta}(\psi,\phi)
  &=
  \mathcal R_{\theta}(\bm Y,H)
  -
  \mathcal R_{\theta}(B_T).
  \label{eq:conditional_ebd}
\end{align}
When the stochastic-degradation relation in \eqref{eq:blackwell_kernel} remains valid conditionally on $\Theta=\theta$, the risk-ordering argument applies statewise and gives
\begin{align}
  \EBD_{\theta}(\psi,\phi)
  &\ge
  0.
  \label{eq:conditional_ebd_nonnegative}
\end{align}
Thus, $\EBD_{\Theta}$ describes how the communication-induced decision loss varies across operating states, while its expectation recovers the corresponding average comparison under the same conditioning convention.

\begin{definition}[Event-belief reliability]
A communication scheme $(\psi,\phi)$ provides $(\delta,\epsilon)$ event-belief reliability if
\begin{align}
  P_{{\rm EB}\text{-}{\rm out}}
  \!\left(
    \delta;
    \psi,\phi
  \right)
  &\triangleq
  \Prob\!\left\{
    \EBD_{\Theta}(\psi,\phi)
    >
    \delta
  \right\}
  \le
  \epsilon,
  \label{eq:eb_reliability}
\end{align}
where $\delta$ is the maximum tolerable state-dependent deficiency and $\epsilon$ is the allowed outage probability.
\end{definition}

The outage probability measures how often representation and wireless delivery cause the state-dependent decision-risk gap to exceed the service tolerance. It therefore combines the declared task, sensing uncertainty, communication design, and operating conditions within a single reliability criterion. 
This reliability requirement leads to the resource-constrained design problem
\begin{align}
  \mathsf{(P_{\rm EB})}:\quad
  \min_{\psi,\phi,P,n}\quad
  &C_{\rm res}(P,n)
  \nonumber\\
  \mathrm{s.t.}\quad
  &
  P_{{\rm EB}\text{-}{\rm out}}
  \!\left(
    \delta;
    \psi,\phi
  \right)
  \le
  \epsilon,
  \nonumber\\
  &
  (\psi,\phi)
  \in
  \mathfrak F(P,n),
  \qquad
  n
  \le
  n_{\max},
  \label{eq:canonical_eb_design}
\end{align}
where $\mathfrak F(P,n)$ is the set of admissible length-$n$ mappings satisfying the power constraint in \eqref{eq:embedding}. Accordingly, problem~\eqref{eq:canonical_eb_design} establishes a general framework for jointly designing the physics-informed mapping and resource allocation subject to a prescribed event-belief reliability requirement.

%% file: sections/analysis.tex
\section{Performance Analysis of Event-Belief Delivery over Rayleigh Fading}
\label{sec:analysis}
In this section, we analyze event-belief delivery over Rayleigh fading under the Gaussian and mode-separable specializations. We derive the LD-EBD-induced received-SNR threshold, characterize the resulting outage probability and required transmit SNR, and further relate the threshold to radio-resource allocation through the delivered event-mode precisions.

\subsection{Gaussian Event Model and Reference Covariances}

We first specialize the event, sensing, and task models to a Gaussian setting. Let $\bm q$, $\bm x$, and $\bm z$ denote the Gaussian vector representations of the event $Q$, the stacked multimodal observation $X^{1:M}$, and the quantity of interest $Z$, respectively. We assume
\begin{align}
  \bm x
  &=
  \mathbf C_x\bm q+\bm v,
  &
  \bm v
  &\sim
  \N\!\left(
    \mathbf 0,
    \mathbf R_x
  \right),
  \label{eq:lgobs}\\
  \bm z
  &=
  \mathbf C_z\bm q+\bm e,
  &
  \bm e
  &\sim
  \N\!\left(
    \mathbf 0,
    \mathbf R_z
  \right),
  \label{eq:lgrep}
\end{align}
where $\bm q$, $\bm v$, and $\bm e$ are mutually independent,
$\boldsymbol{\Sigma}_{\bm q}\succ\mathbf 0$,
$\mathbf R_x\succ\mathbf 0$, and
$\mathbf R_z\succ\mathbf 0$. This specialization is consistent with classical Gaussian indirect-source models~\cite{oohama2014indirect}. For a realized multimodal observation $\bm x$, the Tx-A belief is
\begin{align}
  B_T
  &=
  P_{\bm q|\bm x}(\cdot|\bm x)
  =
  \N\!\left(
    \hat{\bm q}_T,
    \boldsymbol{\Sigma}_{\bm q|\bm x}
  \right),
  \label{eq:gaussian_tx_belief}
\end{align}
where $\hat{\bm q}_T=\E[\bm q|\bm x]$ and
\begin{align}
  \boldsymbol{\Sigma}_{\bm q|\bm x}
  &=
  \left(
    \boldsymbol{\Sigma}_{\bm q}^{-1}
    +
    \mathbf C_x^{T}
    \mathbf R_x^{-1}
    \mathbf C_x
  \right)^{-1}.
  \label{eq:sigqx}
\end{align}
Although the posterior mean $\hat{\bm q}_T$ depends on the specific observation realization, the posterior covariance does not and quantifies the residual event uncertainty after multimodal sensing. The corresponding covariance matrices of $\bm z$ under the prior-only and Tx-A reference beliefs are
\begin{align}
  \boldsymbol{\Sigma}_{\bm z|0}
  &=
  \mathbf C_z
  \boldsymbol{\Sigma}_{\bm q}
  \mathbf C_z^{T}
  +
  \mathbf R_z,
  \label{eq:sigz0}\\
  \boldsymbol{\Sigma}_{\bm z|{\rm ref}}
  &=
  \mathbf C_z
  \boldsymbol{\Sigma}_{\bm q|\bm x}
  \mathbf C_z^{T}
  +
  \mathbf R_z.
  \label{eq:sigzref}
\end{align}
Since
$\boldsymbol{\Sigma}_{\bm q|\bm x}
\preceq
\boldsymbol{\Sigma}_{\bm q}$,
we have
\begin{align}
  \boldsymbol{\Sigma}_{\bm z|{\rm ref}}
  &\preceq
  \boldsymbol{\Sigma}_{\bm z|0}.
  \label{eq:reference_covariance_order}
\end{align}
Thus, $\boldsymbol{\Sigma}_{\bm z|0}$ represents the uncertainty before incorporating the current multimodal observations, whereas
$\boldsymbol{\Sigma}_{\bm z|{\rm ref}}$ represents the residual uncertainty after sensing but before wireless delivery. Under Gaussian log loss, communication-induced degradation can therefore be quantified through the covariance inflation relative to the Tx-A reference.

\subsection{Mode-Separable Belief Delivery}

We specialize the channel model in \eqref{eq:channel_model} to a coherent block-Rayleigh fading channel with $H\sim\CN(0,\beta)$. The channel coefficient remains constant over the $n$ channel uses within each fading block, and the channel state information (CSI) assumptions follow those specified in Section~\ref{sec:system_model}. Define the transmit SNR as $\rho=P/N_0$, the channel power gain as $G=|H|^2$, and the instantaneous received SNR as $\Gamma=\rho G$. Since $H\sim\CN(0,\beta)$, $G$ is exponentially distributed with mean $\beta$. At a realized received SNR $s$, the Rx-A belief is Gaussian:
\begin{align}
  B_R
  &=
  \N\!\left(
    \hat{\bm q}_R,
    \boldsymbol{\Sigma}_{\bm q|R}(s)
  \right),
  \label{eq:gaussian_rx_belief}
\end{align}
which induces
\begin{align}
  \boldsymbol{\Sigma}_{\bm z|R}(s)
  &=
  \mathbf C_z
  \boldsymbol{\Sigma}_{\bm q|R}(s)
  \mathbf C_z^T
  +
  \mathbf R_z.
  \label{eq:sigzR}
\end{align}
The dependence on $s$ is governed by the representation, encoding, and receiver mappings. We model this dependence through a mode-separable covariance structure. Let $r$ denote the dimension of $\bm z$.

\begin{assumption}[Mode-separable covariance structure]
\label{ass:mode_separable}
The prior-only, Tx-A reference, and Rx-A covariance matrices of $\bm z$ are simultaneously diagonalizable in a fixed orthonormal basis. We refer to the coordinates in this common basis as event modes. In this basis,
\begin{align}
  \boldsymbol{\Sigma}_{\bm z|0}
  &=
  \diag\!\left(
    \frac{1}{a_1},
    \ldots,
    \frac{1}{a_r}
  \right),
  \label{eq:diag0}\\
  \boldsymbol{\Sigma}_{\bm z|{\rm ref}}
  &=
  \diag\!\left(
    \frac{1}{a_1+b_{1,{\rm ref}}},
    \ldots,
    \frac{1}{a_r+b_{r,{\rm ref}}}
  \right),
  \label{eq:diagref}\\
  \boldsymbol{\Sigma}_{\bm z|R}(s)
  &=
  \diag\!\left(
    \frac{1}{a_1+\gamma_1(s)},
    \ldots,
    \frac{1}{a_r+\gamma_r(s)}
  \right),
  \label{eq:diagR}
\end{align}
where
\begin{align}
  a_i
  &>
  0,
  &
  0
  &\le
  \gamma_i(s)
  \le
  b_{i,{\rm ref}},
  \qquad
  b_{i,{\rm ref}}
  \ge
  0.
  \label{eq:delivered_precision_bounds}
\end{align}
\end{assumption}

This structure arises naturally when the sensing, representation, and delivery mappings are aligned with a common task-relevant basis. For event mode $i$, $a_i$ is the prior precision, $b_{i,{\rm ref}}$ is the sensing-induced precision increment available at the Tx-A, and $\gamma_i(s)$ is the portion of this increment preserved at the Rx-A. The bounds in \eqref{eq:delivered_precision_bounds} imply
\begin{align}
  \boldsymbol{\Sigma}_{\bm z|{\rm ref}}
  &\preceq
  \boldsymbol{\Sigma}_{\bm z|R}(s)
  \preceq
  \boldsymbol{\Sigma}_{\bm z|0},
  \label{eq:mode_covariance_order}
\end{align}
which is consistent with the communication-induced risk ordering established in Proposition~\ref{prop:risk_ordering}. For tractable threshold analysis, we adopt the hard-ceiling precision law
\begin{align}
  \gamma_i^{\rm H}(s)
  &=
  \min\!\left\{
    b_{i,{\rm ref}},
    b_i s
  \right\},
  \qquad
  b_i>0.
  \label{eq:hard_ceiling_precision}
\end{align}
The coefficient $b_i$ characterizes the delivery efficiency associated with the prescribed representation and communication design. Accordingly, the delivered precision increases linearly with the received SNR and saturates at the sensing-limited ceiling $b_{i,{\rm ref}}$. 

\subsection{Log-Det EBD and Received-SNR Threshold}

For a fixed transmit SNR $\rho$, let the channel power gain $G$ serve as the operating state in \eqref{eq:conditional_ebd}. At $G=g$, the instantaneous received SNR is $s=\rho g$. The Gaussian specialization of the state-dependent log-loss EBD is then characterized by the corresponding posterior covariance matrices.

\begin{proposition}[Gaussian log-det representation]
\label{prop:gaussian_logdet_ebd}
For any realized received SNR $s\ge 0$, the state-dependent log-loss EBD is
\begin{align}
  \LDEBD(s)
  &=
  \frac{1}{2}
  \log
  \frac{
    \det\!\left(
      \boldsymbol{\Sigma}_{\bm z|R}(s)
    \right)
  }{
    \det\!\left(
      \boldsymbol{\Sigma}_{\bm z|{\rm ref}}
    \right)
  }.
  \label{eq:LDEBDmatrix}
\end{align}
\end{proposition}

\begin{proof}
By Corollary~\ref{cor:log_ebd}, the log-loss EBD equals the difference between the conditional differential entropies of $\bm z$ at the Rx-A and Tx-A. For an $r$-dimensional Gaussian vector, the conditional differential entropy is one half of the logarithm of the determinant of its conditional covariance, up to the common additive term $r\log(2\pi e)/2$, which cancels in the difference.
\end{proof}

Since all logarithms are natural, $\LDEBD(s)$ is measured in nats and quantifies the logarithmic inflation in posterior uncertainty volume at the Rx-A relative to the Tx-A reference. Under Assumption~\ref{ass:mode_separable}, this covariance-level measure admits the following mode-wise scalar decomposition:
\begin{align}
  \LDEBD(s)
  &=
  \frac{1}{2}
  \sum_{i=1}^{r}
  \log
  \frac{
    a_i+b_{i,{\rm ref}}
  }{
    a_i+\gamma_i(s)
  }.
  \label{eq:LDEBDmodes_general}
\end{align}
Each summand quantifies the residual information loss associated with an individual event mode and is nonnegative by \eqref{eq:delivered_precision_bounds}. With the event model, sensing configuration, and delivery design held fixed, we henceforth write
\begin{align}
  \mathcal D(s)
  &=
  \LDEBD(s)
  \label{eq:ldebd_scalar_function}
\end{align}
for the resulting received-SNR response. For a design that asymptotically delivers every mode satisfying $b_{i,{\rm ref}}>0$,
\begin{align}
  \gamma_i(0) = 0,
  \quad
  \lim_{s\to\infty} \gamma_i(s) = b_{i,{\rm ref}}.
  \label{eq:precision_endpoints}
\end{align}
The maximum LD-EBD is therefore attained at zero received SNR and is given by
\begin{align}
  \delta_{\max} = \mathcal D(0)
  = \frac{1}{2} \sum_{i=1}^{r} \log \frac{ a_i+b_{i,{\rm ref}} }{ a_i}.
  \label{eq:deltamaxgeneral}
\end{align}

Suppose that $\mathcal D(s)$ is continuous and strictly decreasing over the region in which $\mathcal D(s)>0$. For any feasible target satisfying
\begin{align}
  0
  &<
  \delta
  <
  \delta_{\max},
  \label{eq:feasrange}
\end{align}
define the EBD-induced received-SNR threshold $\kappa_\delta$ as the unique solution to
\begin{align}
  \mathcal D(\kappa_\delta)
  &=
  \delta.
  \label{eq:kappadef}
\end{align}
The threshold $\kappa_\delta$ specifies the minimum instantaneous received SNR required to satisfy the prescribed LD-EBD target. It is determined by the Gaussian event and sensing parameters together with the delivery design, but is independent of the fading distribution. Under the hard-ceiling model, the received-SNR response becomes
\begin{align}
  \mathcal D_{\rm H}(s)
  &=
  \frac{1}{2}
  \sum_{i=1}^{r}
  \log
  \frac{
    a_i+b_{i,{\rm ref}}
  }{
    a_i+
    \min\!\left\{
      b_{i,{\rm ref}},
      b_i s
    \right\}
  }.
  \label{eq:LDEBDmodes}
\end{align}
The corresponding threshold $\kappa_{\delta,{\rm H}}$ is the unique solution to
\begin{align}
  \prod_{i=1}^{r}
  \frac{
    a_i+b_{i,{\rm ref}}
  }{
    a_i+
    \min\!\left\{
      b_{i,{\rm ref}},
      b_i\kappa_{\delta,{\rm H}}
    \right\}
  }
  &=
  e^{2\delta}.
  \label{eq:kappaproduct}
\end{align}
The left-hand side is continuous and strictly decreasing over the relevant interval, and hence $\kappa_{\delta,{\rm H}}$ can be efficiently computed by a one-dimensional bisection search. The received-SNR threshold thereby provides the interface between the event-belief reliability requirement and channel fading.

\subsection{Rayleigh Outage and Link-Budget Scaling}

For a fixed event model, sensing configuration, and delivery design, the LD-EBD target is violated if and only if $\Gamma=\rho G$ falls below $\kappa_\delta$. Hence,
\begin{align}
  P_{{\rm EB}\text{-}{\rm out}}(\delta;\rho)
  &=
  \Prob\!\left\{
    \mathcal D(\rho G)>\delta
  \right\}
  \nonumber\\
  &=
  \Prob\!\left\{
    \rho G<\kappa_\delta
  \right\},
  \label{eq:eboutdef}
\end{align}
where the dependence on the fixed model, sensing, and delivery parameters is suppressed for notational simplicity. By the definition of the required received-SNR threshold, the reliability requirement in \eqref{eq:eb_reliability} is equivalently expressed as
\begin{align}
  P_{{\rm EB}\text{-}{\rm out}}(\delta;\rho)
  &\le
  \epsilon.
  \label{eq:rayleigh_reliability_requirement}
\end{align}

\begin{theorem}[Inverse-SNR scaling of event-belief outage]
\label{thm:siso_event_belief_outage}
For the considered immersive communication system, the event-belief outage probability decreases as either the transmit SNR or the average channel power gain increases. In particular, at high SNR, the outage probability vanishes inversely with the transmit SNR and therefore achieves first-order diversity. Equivalently, the system achieves a diversity order of one. More precisely, for $H\sim\CN(0,\beta)$, $\rho>0$, and every
$0<\delta<\delta_{\max}$,
\begin{align}
  P_{{\rm EB}\text{-}{\rm out}}(\delta;\rho)
  &=
  1-
  \exp\!\left(
    -\frac{\kappa_\delta}{\beta\rho}
  \right).
  \label{eq:outagelaw}
\end{align}
\end{theorem}

\begin{proof}
See Appendix~\ref{app:proof_siso_event_belief_outage}.
\end{proof}

The theorem establishes a separation between threshold determination and fading-induced outage: the Gaussian event and sensing parameters, the delivery design, and the prescribed LD-EBD target jointly determine $\kappa_\delta$, whereas the Rayleigh fading law maps this deterministic threshold to the corresponding outage probability. Accordingly, inverting \eqref{eq:outagelaw} yields the minimum transmit SNR required to satisfy a prescribed outage constraint $\epsilon\in(0,1)$:
\begin{align}
  \rho_{\rm req}(\delta,\epsilon)
  &=
  \frac{
    \kappa_\delta
  }{
    \beta
    \left[
      -\log(1-\epsilon)
    \right]
  }.
  \label{eq:requiredsnr}
\end{align}
Thus, any reduction in $\kappa_\delta$ translates directly into a proportional reduction in the required link budget. At high transmit SNR,
\begin{align}
  &P_{{\rm EB}\text{-}{\rm out}}(\delta;\rho)
  =
  \frac{
    \kappa_\delta
  }{
    \beta
  }
  \rho^{-1}
  +
  o\!\left(
    \rho^{-1}
  \right),
  \label{eq:highsnr}\\
  &d_{\rm EB}(\delta)
  =
  -
  \lim_{\rho\to\infty}
  \frac{
    \log
    P_{{\rm EB}\text{-}{\rm out}}(\delta;\rho)
  }{
    \log\rho
  }
  =
  1.
  \label{eq:divorder}
\end{align}
Hence, the event, sensing, and delivery parameters determine the high-SNR outage prefactor through $\kappa_\delta$, whereas the Rayleigh channel fixes the diversity order at one. 

\subsection{Resource Coupling for Event-Native Delivery}
The hard-ceiling model characterizes the received-SNR threshold in terms of fixed delivery slopes $\{b_i\}$. We now relate the delivered precisions to radio resources through an orthogonal Gaussian transmission scheme aligned with the common event-mode basis. Define the set of active event modes as
\begin{align}
  \mathcal M_{\rm act}
  &=
  \left\{
    i:
    b_{i,{\rm ref}}>0
  \right\}.
  \label{eq:active_mode_set}
\end{align}
For each $i\in\mathcal M_{\rm act}$, let $z_i$ denote the corresponding coordinate of $\bm z$. Since its prior and Tx-A posterior precisions are $a_i$ and $a_i+b_{i,{\rm ref}}$, respectively, Gaussian precision matching yields an equivalent scalar statistic $T_i=T_i(B_T)$ satisfying
\begin{align}
  T_i
  &\overset{\rm d}{=}
  z_i+N_{T,i},
  &
  N_{T,i}
  &\sim
  \N\!\left(
    0,
    b_{i,{\rm ref}}^{-1}
  \right),
  \label{eq:equivalent_source_observation}
\end{align}
where $N_{T,i}$ is independent of $z_i$. This posterior-equivalent representation does not require the Tx-A to observe $z_i$ directly. The variance of the equivalent observation is
\begin{align}
  \nu_i
  =
  \E[T_i^2]
  =
  \frac{1}{a_i}
  +
  \frac{1}{b_{i,{\rm ref}}}.
  \label{eq:equivalent_observation_variance}
\end{align}

Let $n_i$ channel uses and per-use power $P_i$ be allocated to event mode $i$, subject to
\begin{align}
  \sum_{i\in\mathcal M_{\rm act}}
  n_i
  \le
  n,
  \quad
  \sum_{i\in\mathcal M_{\rm act}}
  n_iP_i
  \le
  nP.
  \label{eq:mode_resource_constraints}
\end{align}
Define the corresponding energy-allocation fractions as
\begin{align}
  \lambda_i
  =
  \frac{
    n_iP_i
  }{
    nP
  },
  \quad
  \lambda_i
  \ge
  0,
  \quad
  \sum_{i\in\mathcal M_{\rm act}}
  \lambda_i
  \le
  1.
  \label{eq:mode_energy_fraction}
\end{align}

After variance normalization, transmission over the assigned channel uses, and coherent combining, the Rx-A obtains the equivalent observation
\begin{align}
  &\widetilde T_i
  =
  z_i
  +
  N_{T,i}
  +
  N_{C,i},
  \label{eq:equivalent_channel_observation}\\
  &N_{C,i}
  \sim
  \N\!\left(
    0,
    \bigl(c_i s\bigr)^{-1}
  \right),
  \quad
  c_i
  =
  \frac{
    n_iP_i
  }{
    P\nu_i
  }
  =
  \frac{
    n\lambda_i
  }{
    \nu_i
  },
  \label{eq:resource_to_precision}
\end{align}
where $N_{C,i}$ is independent of $(z_i,N_{T,i})$. The coefficient $c_i$ is the communication-induced precision slope of event mode $i$. Since the sensing and communication noises are independent, the delivered precision increment is
\begin{align}
  \gamma_i^{\rm G}
  \!\left(
    s;\bm\lambda
  \right)
  &=
  \frac{
    b_{i,{\rm ref}}c_i s
  }{
    b_{i,{\rm ref}}+c_i s
  },
  \qquad
  i\in\mathcal M_{\rm act}.
  \label{eq:smooth_delivered_precision}
\end{align}

When $b_i=c_i$, this precision law has the same low-SNR slope and sensing-limited ceiling as the hard-ceiling model. Moreover,
\begin{align}
  \frac{1}{2}
  \min\!\left\{
    b_{i,{\rm ref}},
    c_i s
  \right\}
  &\le
  \gamma_i^{\rm G}
  \!\left(
    s;\bm\lambda
  \right)
  \le
  \min\!\left\{
    b_{i,{\rm ref}},
    c_i s
  \right\}.
  \label{eq:harmonic_min_bound}
\end{align}
Under the matched choice $b_i=c_i$, \eqref{eq:harmonic_min_bound} further implies
$\mathcal D_{\rm H}(s)\le\mathcal D_{\rm G}(s;\bm\lambda)$ and
$\kappa_{\delta,{\rm H}}\le\kappa_{\delta,{\rm G}}(\bm\lambda)$ for any common feasible target. The hard-ceiling model therefore provides an optimistic but tractable benchmark, while the Gaussian scheme exposes the dependence on radio-resource allocation.

Substituting \eqref{eq:smooth_delivered_precision} into
\eqref{eq:LDEBDmodes_general} yields the scheme-specific LD-EBD response
\begin{align}
  \mathcal D_{\rm G}
  \!\left(
    s;\bm\lambda
  \right)
  &=
  \frac{1}{2}
  \sum_{i\in\mathcal M_{\rm act}}
  \log
  \frac{
    a_i+b_{i,{\rm ref}}
  }{
    a_i+
    \gamma_i^{\rm G}
    \!\left(
      s;\bm\lambda
    \right)
  }.
  \label{eq:gaussian_scheme_ldebd}
\end{align}
For any fixed allocation $\bm\lambda$ and feasible target satisfying
\begin{align}
  \mathcal D_{\rm G}(\infty;\bm\lambda)
  &<
  \delta
  <
  \mathcal D_{\rm G}(0;\bm\lambda),
  \label{eq:gaussian_scheme_feasible_target}
\end{align}
define the corresponding received-SNR threshold $\kappa_{\delta,{\rm G}}(\bm\lambda)$ through
\begin{align}
  \mathcal D_{\rm G}
  \!\left(
    \kappa_{\delta,{\rm G}}(\bm\lambda);
    \bm\lambda
  \right)
  &=
  \delta.
  \label{eq:gaussian_scheme_threshold}
\end{align}
Combined with Theorem~\ref{thm:siso_event_belief_outage}, this threshold characterization provides the basis for target-dependent event-native radio-resource allocation.

%% file: sections/simulation.tex
\section{Numerical Results}
\label{sec:numerical_results}

This section presents a numerical evaluation of the proposed event-native immersive-communication framework in a physics-based fluid-interaction (FIT) environment. In this section, we first assess the validity of stream distortion as a surrogate for the receiver's knowledge of the physical event. We next compare the amount of sensing-supported event information retained by the event-native and stream-oriented source interfaces under the same transmission constraints. We finally verify the analytical event-belief outage characterization and demonstrate that the event-native belief advantage yields predictable reliability gains and a reduced average received-SNR requirement.

\subsection{Physics-Based Environment and Belief Calibration}

We consider a FIT event in which a moving contact jointly excites a membrane displacement field $u(x,y,t)$ and a pressure field $p(x,y,t)$~\cite{pastewka2014roughContact,sahli2018realContactArea}. As an instance of \eqref{eq:physical_source_pde}, the coupled dynamics are
\begin{align}
  &\rho_s u_{tt}
  +
  c_s u_t
  +
  k_s u
  -
  T_s\nabla^2u
  =
  F_{\rm contact}(x,y,t)
  -
  \beta_p p,
  \label{eq:fit_membrane}\\
  &\tau_f p_t
  =
  D_f\nabla^2p
  -
  p
  +
  \kappa_f u_t,
  \label{eq:fit_pressure}
\end{align}
where $F_{\rm contact}(x,y,t)$ is a moving load centered at $(x_c(t),y_c(t))$. A single contact episode therefore produces deformation, pressure diffusion, vibration, and audio-related observations. These modalities share a common physical cause, but their event-relevant components need not be aligned by waveform shape, dominant frequency, or stream-domain energy. To obtain a finite-dimensional event representation, we project $u(x,y,t)$ and $p(x,y,t)$ onto their dominant spatial modes, with coefficients $\bm\eta_k$ and $\bm\xi_k$, respectively. The resulting FIT event coordinate is
\begin{align}
  \bm q_k^{\rm FIT}
  &=
  \left[
    \bm\eta_k^T,
    \dot{\bm\eta}_k^T,
    \bm\xi_k^T,
    x_{c,k},
    y_{c,k},
    v_{c,k},
    F_{n,k},
    F_{t,k}
  \right]^T.
  \label{eq:fit_event_coordinate}
\end{align}
Following \eqref{eq:sensor_operator}, modality-specific operators extract contact, kinematic, pressure, vibration, and audio features from each FIT episode. Fig.~\ref{fig:fit_source} illustrates this event-to-observation mapping, showing that the modalities are coupled through the latent event rather than through identical stream coordinates.

\begin{figure*}[t]
  \centering
  \includegraphics[width=0.49\textwidth]{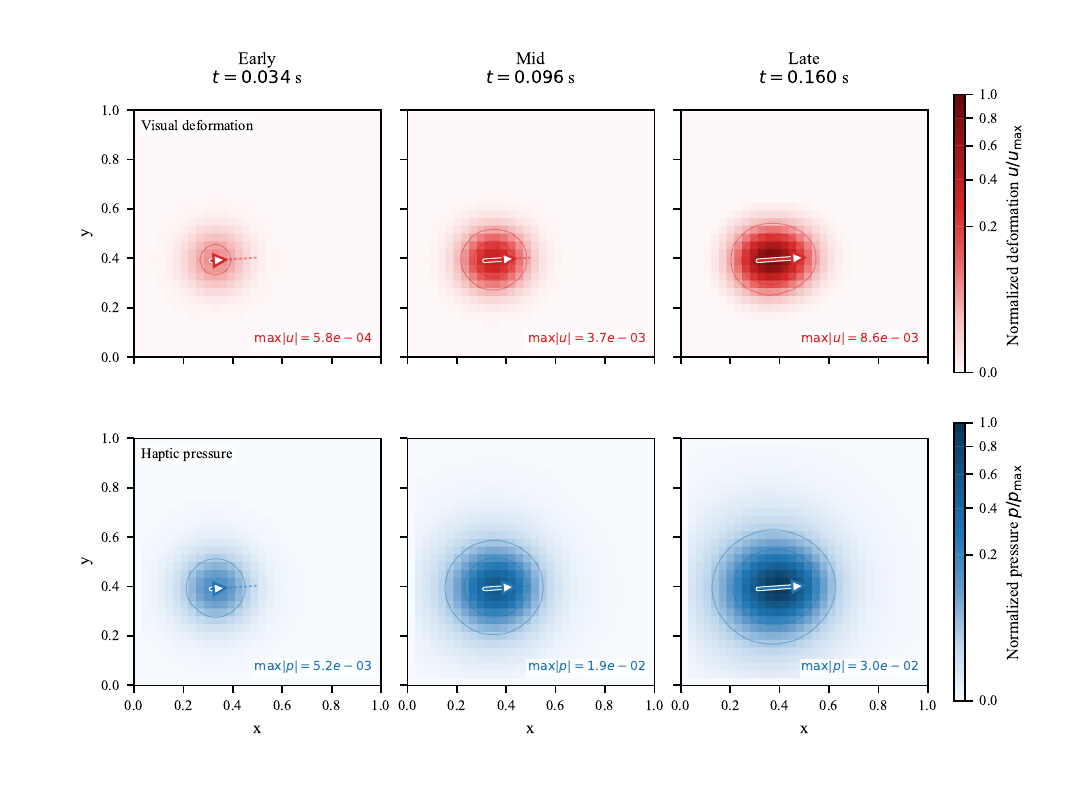}
  \hfill
  \includegraphics[width=0.49\textwidth]{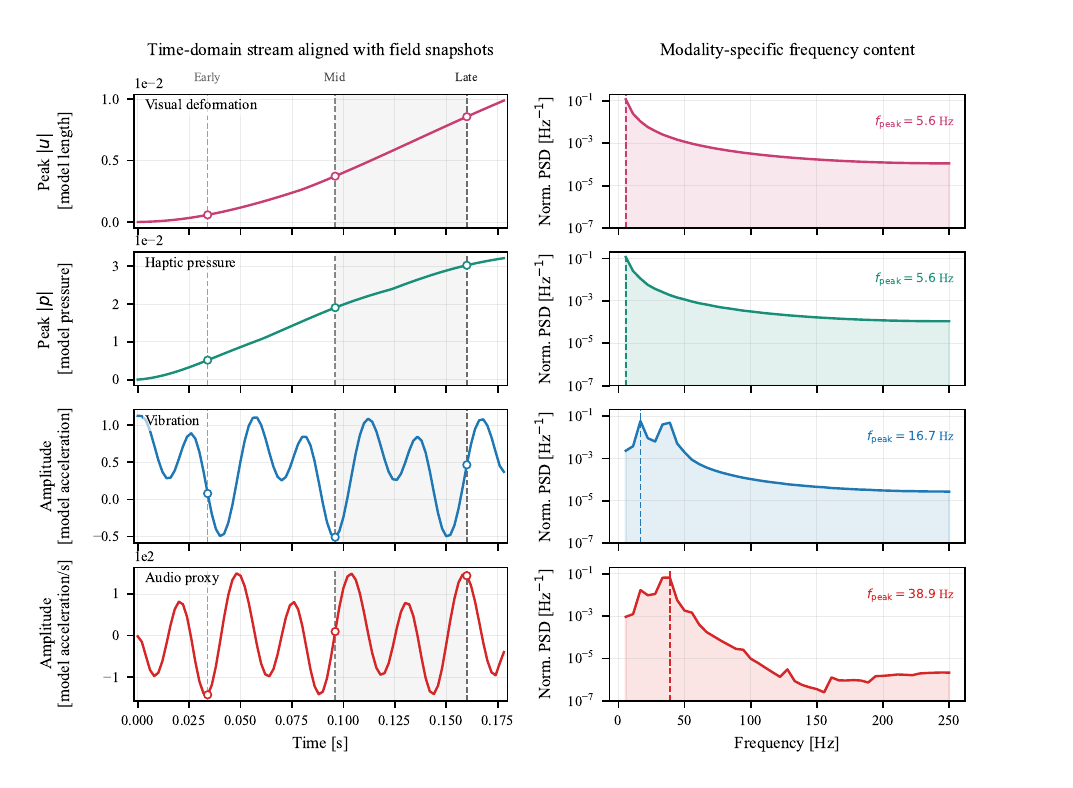}
  \caption{FIT event source and multimodal observations. A common moving-contact episode induces membrane deformation, pressure diffusion, vibration, and audio-related features whose informative components are aligned through the latent physical event rather than stream-domain energy.}
  \label{fig:fit_source}
\end{figure*}

The nonlinear simulator generates $10^6$ event--observation pairs $(\bm q_k^{\rm FIT},\bm x_k)$. After centering, we calibrate the Gaussian observation model in \eqref{eq:lgobs} from these pairs, obtaining the event covariance, observation operator, and residual covariance that determine the Tx-A belief $B_T$. We then express the sensing-supported uncertainty in a three-mode event basis associated with contact, pressure, and vibration information, and use the resulting mode parameters to instantiate the Rx-A covariance model in Section~\ref{sec:analysis}. 

The calibrated beliefs are delivered over the coherent Rayleigh channel in Section~\ref{sec:analysis}. The propagation and noise parameters in Table~\ref{tab:wireless_parameters} map the transmit-power sweep to the average received SNR
$\bar\Gamma=\E[\Gamma]=\beta\rho$.
We compare four source interfaces under identical transmitted dimension, blocklength, energy budget, and fading realizations. Event-native preserves the event modes that define the Tx-A belief; MSE-oriented minimizes multimodal-observation reconstruction error; Stream principal component analysis (PCA) retains the dominant observation-variance directions; and Uncoded analog directly transmits selected observation coordinates.

Performance is reported using normalized LD-EBD in \eqref{eq:normalizedEBD}. For the fixed event model, sensing model, and task considered here, $D_0-D_{\rm ref}$ is common to all interfaces. Normalization therefore rescales the LD-EBD threshold without changing the interface ordering or the threshold--outage relation derived in Section~\ref{sec:analysis}. An auxiliary normalized physical-state estimation error is included only to aid interpretation.

\begin{table}[t]
  \centering
  \small
  \caption{Wireless-link and evaluation parameters.}
  \label{tab:wireless_parameters}
  \setlength{\tabcolsep}{4pt}
  \renewcommand{\arraystretch}{1.08}
  \begin{tabular}{@{}lc@{}}
    \toprule
    Parameter & Value\\
    \midrule
    Carrier frequency $f_c$ & $3.5$~GHz\\
    Bandwidth & $20$~MHz\\
    Link distance & $500$~m\\
    Tx/Rx heights & $10/1.5$~m\\
    Tx/Rx gains & $5/0$~dBi\\
    Noise figure & $7$~dB\\
    $P_{\rm Tx}$ sweep & $-5{:}5{:}40$~dBm\\
    Channel realizations / seed & $10^6 / 7$\\
    \bottomrule
  \end{tabular}
\end{table}

\subsection{Metric Validity and Event-Native Delivery}

We first test whether stream distortion can serve as a surrogate for receiver knowledge of the physical event. Fig.~\ref{fig:metric_validity} compares physically distinct FIT procedures whose stream NMSE values strongly overlap. The median NMSE changes by only $1.38\%$, whereas the median LD-EBD changes by a factor of $17.32$. Their Spearman correlations with the reference event-risk ordering are $-0.035$ and $0.995$, respectively. Stream NMSE therefore fails to distinguish procedures that induce substantially different posterior uncertainty, while LD-EBD preserves their event-risk ordering.

\begin{figure}[t]
  \centering
  \includegraphics[width=\columnwidth]{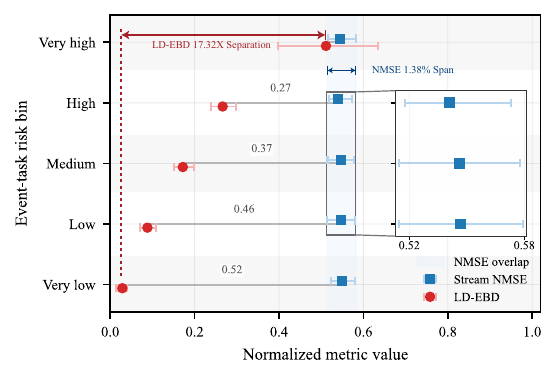}
  \caption{Metric validity for receiver event knowledge. Physically distinct FIT procedures have overlapping stream NMSE values, whereas LD-EBD separates their event-risk levels.}
  \label{fig:metric_validity}
\end{figure}

Having established that NMSE is not a reliable surrogate, we use normalized LD-EBD to compare the four source interfaces under common channel realizations. Fig.~\ref{fig:ld_ebd_snr} shows that Event-native achieves the smallest normalized LD-EBD throughout the simulated average-received-SNR range. Relative to MSE-oriented, its pointwise reduction increases from $3.1\%$ at low SNR to $63.6\%$ at high SNR, with an average reduction of $37.6\%$.

\begin{figure}[t]
  \centering
  \includegraphics[width=\columnwidth]{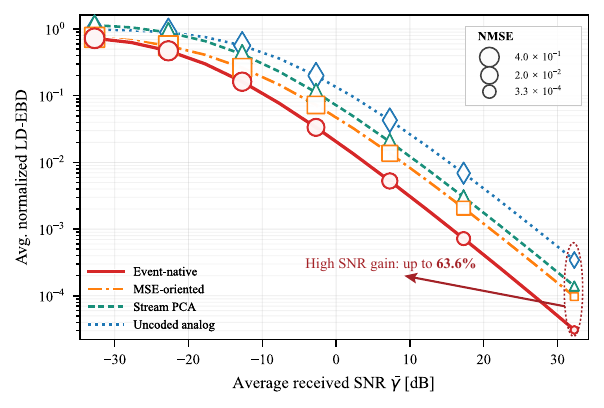}
  \caption{Normalized LD-EBD versus average received SNR $\bar\Gamma$ for the four source interfaces. Bubble sizes indicate the auxiliary normalized physical-state estimation error.}
  \label{fig:ld_ebd_snr}
\end{figure}

The SNR dependence clarifies the source of this gain. At low $\bar\Gamma$, channel noise limits all interfaces and reduces the effect of source-interface choice. As $\bar\Gamma$ increases, communication noise recedes and the dominant difference becomes how each interface allocates its representation to sensing-supported event information. Event-native then retains a larger fraction of the Tx-A decision capability, while stream-oriented interfaces continue to spend resources on directions selected by observation energy or reconstruction error.

The bubble sizes in Fig.~\ref{fig:ld_ebd_snr} provide demonstration of NMSE under different SNR. The LD-EBD reduction is accompanied by competitive physical-state estimation accuracy, indicating that the event-level gain is not obtained by discarding receiver-relevant physical information. The next question is whether the analytical received-SNR threshold accurately predicts the reliability of this gain under random fading.

\subsection{Reliability and Resource-Design Gains}

Fig.~\ref{fig:outage} compares the analytical event-belief outage probability in \eqref{eq:outagelaw} with Monte Carlo estimates for several normalized LD-EBD thresholds. The maximum absolute deviation is $6.93\times10^{-4}$, and the maximum relative deviation is $2.27\%$ for outage probabilities above $10^{-4}$. The agreement validates the threshold--fading separation in Theorem~\ref{thm:siso_event_belief_outage}: the calibrated event, sensing, and interface parameters determine the received-SNR threshold, while the Rayleigh channel maps that threshold into outage probability.

\begin{figure}[t]
  \centering
  \includegraphics[width=\columnwidth]{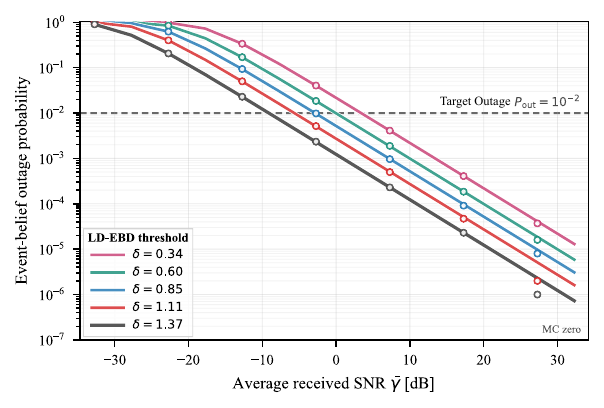}
  \caption{Validation of the event-belief outage law under Rayleigh fading. Solid curves show \eqref{eq:outagelaw}, and markers show Monte Carlo estimates for different normalized LD-EBD thresholds.}
  \label{fig:outage}
\end{figure}

The validated outage law allows the belief advantage to be translated into a communication-resource requirement. For normalized LD-EBD target $0.5$ and allowed outage probability $10^{-2}$, Fig.~\ref{fig:required_snr} shows that Event-native requires $-3.00$~dB average received SNR over 16 bootstrap recalibrations. MSE-oriented requires $0.48$~dB under the same condition, giving Event-native a $3.47$~dB saving. Stream PCA and Uncoded analog require still larger average received SNR. Because all interfaces use the same fading realizations and link budget, these differences are source-interface gains rather than channel-coding gains. The bootstrap spread further shows that the required-SNR advantage is stable under recalibration of the FIT belief model.

\begin{figure}[t]
  \centering
  \includegraphics[width=\columnwidth]{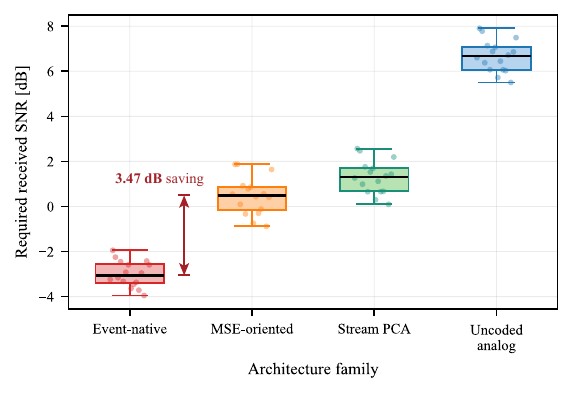}
  \caption{Required average received SNR under a normalized LD-EBD outage constraint. Error bars summarize 16 bootstrap recalibrations of the FIT belief model.}
  \label{fig:required_snr}
\end{figure}

We finally examine how the internal event-mode allocation of the Event-native interface should vary with the reliability target. The considered allocation family is indexed by a concentration parameter $\alpha$, with larger $\alpha$ producing a more concentrated allocation across event modes. Fig.~\ref{fig:event_native_tradeoff} reports the required-SNR penalty over this family and the best concentration parameter $\alpha^\star$ for each LD-EBD target. Here, $\alpha^\star$ denotes the best value within the considered parameterized family rather than a globally optimal allocation.

\begin{figure*}[t]
  \centering
  \includegraphics[width=0.49\textwidth]{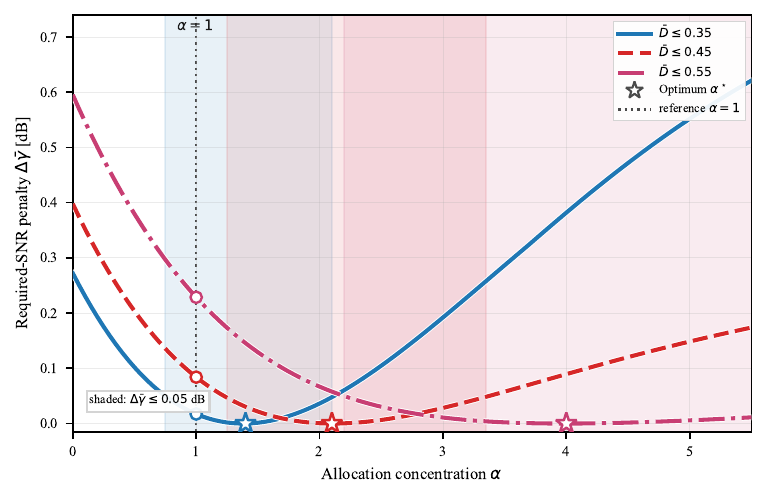}
  \hfill
  \includegraphics[width=0.49\textwidth]{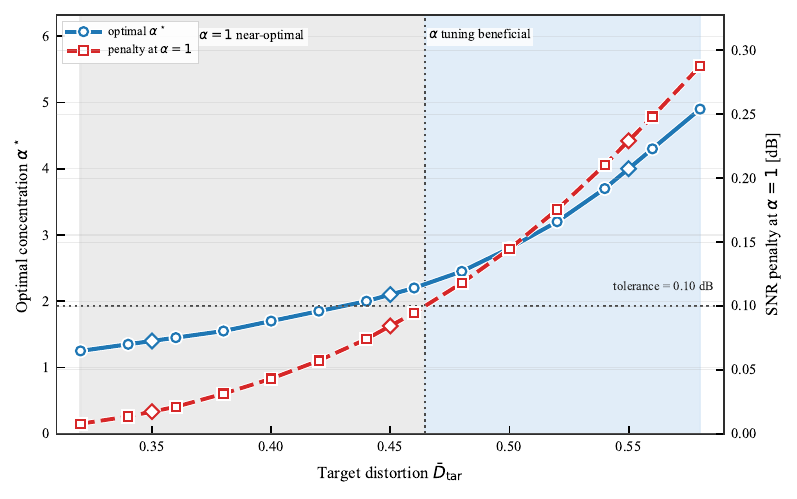}
  \caption{Event-native allocation tradeoff. Left: required-SNR penalty versus the concentration parameter $\alpha$. Right: the best $\alpha^\star$ within the considered family and the penalty of the reference design $\alpha=1$ versus the normalized LD-EBD target.}
  \label{fig:event_native_tradeoff}
\end{figure*}

As the admissible LD-EBD target is relaxed, the minimizing concentration parameter increases, indicating that a more concentrated event-mode allocation becomes preferable. The reference design $\alpha=1$ remains near-optimal in the stricter-target regime, while tuning becomes more beneficial for relaxed targets. The small reference-design penalty shows that the event-native interface is robust to moderate mistuning; the shift in $\alpha^\star$ shows that the preferred resource concentration still depends on the declared event-belief reliability target. EBD therefore serves not only as an evaluation metric but also as a target-dependent design criterion for 6G immersive communication.

%% file: sections/conclusion.tex
\section{Conclusion}
\label{sec:conclusion}

This paper develops an event-native foundation for immersive communication by making a service-relevant physical event, rather than its sensory streams, the communication object. The Tx-A and Rx-A beliefs provide the corresponding operational information states, while event-belief deficiency measures the communication-induced loss of event-dependent decision capability relative to the sensing-supported Tx-A reference. For Gaussian event beliefs with mode-separable delivery over Rayleigh fading, this fidelity criterion induces a received-SNR threshold that yields an exact outage expression and a required-SNR characterization, thereby connecting physical-event preservation to wireless reliability. Experiments in a physics-based environment showed that stream NMSE can remain nearly unchanged across physically distinct event procedures, whereas LD-EBD accurately reflects their event-level uncertainty. The analytical outage law closely matched Monte Carlo results, and, relative to an MSE-oriented interface, event-native delivery reduced normalized LD-EBD by up to $63.6\%$ and lowered the required average received SNR by $3.47$~dB under the evaluated outage constraint. These findings support a formulation of immersive communication based on preserving physical-event beliefs rather than reconstructing individual modality streams, and motivate future work on nonlinear and learned event representations, interactive multi-agent systems, and experimental prototypes.

%% file: sections/appendix.tex
\appendices

\section{Proof of Lemma~\ref{prop:belief_sufficiency}}
\label{app:proof_belief_sufficiency}

Let $\cI_T=X^{1:M}$. For every measurable set $A\subseteq\cZ$, the conditional independence in \eqref{eq:zx_condind} gives
\begin{align}
  P_{Z|\cI_T}(A|\cI_T)
  =
  \int_{\cQ}
  P_{\rm imm}(A|q)
  B_T(dq)
  \triangleq
  \kappa_A(B_T),
  \label{eq:app_predictive_kernel}
\end{align}
where $\kappa_A(B_T)$ is a measurable function of the Tx-A belief. Since $B_T$ is itself determined by $\cI_T$, the tower property yields
\begin{align}
  P_{Z|B_T}(A|B_T)
  &=
  \E\!\left[
    P_{Z|\cI_T}(A|\cI_T)
    \,\middle|\,
    B_T
  \right]
  \nonumber\\
  &=
  \E\!\left[
    \kappa_A(B_T)
    \,\middle|\,
    B_T
  \right]
  \nonumber\\
  &=
  \kappa_A(B_T)
  \nonumber\\
  &=
  P_{Z|\cI_T}(A|\cI_T).
  \label{eq:app_belief_sufficiency}
\end{align}
Therefore,
\begin{align}
  Z
  &\mathrel{\perp\!\!\!\perp}
  X^{1:M}
  \mid
  B_T.
  \label{eq:app_conditional_independence}
\end{align}
The conditional law relevant to every Bayesian decision problem based on $Z$ is thus determined by $B_T$, so an optimal rule based on $X^{1:M}$ can be expressed as a function of the Tx-A belief.

\section{Proof of Proposition~\ref{prop:stream_event_nonequivalence}}
\label{app:proof_stream_event_nonequivalence}

Let $Q_1,Q_2\sim\N(0,1)$ be independent, define
\begin{align}
  X
  =
  \left[
    Q_1,
    Q_2
  \right]^T,
  \quad
  Z
  =
  Q_1,
  \label{eq:app_counterexample_variables}
\end{align}
and consider the Rx-A information states
\begin{align}
  \cI_1
  =
  Q_1,
  \quad
  \cI_2
  =
  Q_2.
  \label{eq:app_counterexample_states}
\end{align}
Their minimum mean-squared stream reconstructions are the corresponding conditional means:
\begin{align}
  \E[X|\cI_1]
  =
  \left[
    Q_1,
    0
  \right]^T,
  \quad
  \E[X|\cI_2]
  =
  \left[
    0,
    Q_2
  \right]^T.
  \label{eq:app_stream_estimators}
\end{align}
Hence, the two information states have the same minimum stream-reconstruction MSE,
\begin{align}
  \E\!\left[
    \left\|
      X-\E[X|\cI_1]
    \right\|^2
  \right]
  &=
  1,
  \nonumber\\
  \E\!\left[
    \left\|
      X-\E[X|\cI_2]
    \right\|^2
  \right]
  &=
  1.
  \label{eq:app_equal_stream_mse}
\end{align}
For squared-error loss on $Z$, however,
\begin{align}
  \E\!\left[
    \left(
      Z-\E[Z|\cI_1]
    \right)^2
  \right]
  &=
  0,
  \nonumber\\
  \E\!\left[
    \left(
      Z-\E[Z|\cI_2]
    \right)^2
  \right]
  &=
  1.
  \label{eq:app_different_event_risks}
\end{align}
Thus, equal minimum stream-reconstruction MSE does not imply equal event-dependent decision risk.

\section{Proof of Theorem~\ref{thm:siso_event_belief_outage}}
\label{app:proof_siso_event_belief_outage}

For $0<\delta<\delta_{\max}$, the received-SNR response $\mathcal D(s)$ is strictly decreasing while positive, and $\kappa_\delta$ is the unique solution to \eqref{eq:kappadef}. Therefore,
\begin{align}
  \left\{
    \mathcal D(\rho G)>\delta
  \right\} =
  \left\{
    \rho G<\kappa_\delta
  \right\}=
  \left\{
    G<
    \frac{\kappa_\delta}{\rho}
  \right\}.
  \label{eq:app_outage_event}
\end{align}
Since $H\sim\CN(0,\beta)$, the channel power gain
$G=|H|^2$ is exponentially distributed with mean $\beta$ and cumulative distribution function
\begin{align}
  F_G(g)
  &=
  1-
  \exp\!\left(
    -\frac{g}{\beta}
  \right),
  \qquad
  g\ge0.
  \label{eq:app_rayleigh_cdf}
\end{align}
It follows that
\begin{align}
  P_{{\rm EB}\text{-}{\rm out}}(\delta;\rho)
  &=
  \Prob\!\left\{
    G<
    \frac{\kappa_\delta}{\rho}
  \right\}
  =
  F_G\!\left(
    \frac{\kappa_\delta}{\rho}
  \right)
  \nonumber\\
  &=
  1-
  \exp\!\left(
    -\frac{\kappa_\delta}{\beta\rho}
  \right),
  \label{eq:app_outage_law}
\end{align}
which proves \eqref{eq:outagelaw}.